\documentclass{article}
\usepackage[font=small,labelfont=bf]{caption}
\usepackage{eurosym}
\usepackage{amsmath}
\usepackage{amsfonts}
\usepackage{tabularx}
\usepackage[table]{xcolor}
\usepackage{mathtools}
\usepackage{mathrsfs}
\newtheorem{theorem}{Theorem}

\newtheorem{corollary}[theorem]{Corollary}

\newtheorem{definition}[theorem]{Definition}
\newtheorem{example}[theorem]{Example}

\newtheorem{lemma}[theorem]{Lemma}

\newenvironment{proof}[1][Proof]{\noindent\textbf{#1.} }{\ \rule{0.5em}{0.5em}}

\newcommand{\ZZ}{{\mathbb{Z}_{2}\mathbb{Z}_{2}[u]}}
\newcommand{\C}{{\mathcal{C}}}
\newcommand{\ty}{{(r,s;k_0,k_1,k_2)}}

\newcommand{\Z}{{\mathbb{Z}}}
\newcommand{\m}{{\texttt{m}}}

\begin{document}

\title{The Structure of One Weight Linear and Cyclic Codes Over $\mathbb{Z}_{2}^{r}\times\left(\mathbb{Z}_{2}+u\mathbb{Z}_{2}\right)^{s}$}

\author{Ismail Aydogdu \footnote{ Department of Mathematics, Yildiz Technical University, Email: {iaydogdu@yildiz.edu.tr }}}
{ }
\maketitle

\begin{abstract}
Inspired by the $\Z_{2}\Z_{4}$-additive codes, linear codes over $\mathbb{Z}_{2}^{r}\times\left(\mathbb{Z}_{2}+u\mathbb{Z}_{2}\right)^{s}$ have been introduced by Aydogdu et al. more recently. Although these family of codes are similar to each other, linear codes over $\mathbb{Z}_{2}^{r}\times\left(\mathbb{Z}_{2}+u\mathbb{Z}_{2}\right)^{s}$ have some advantages compared to $\Z_{2}\Z_{4}$-additive codes. A code is called constant weight(one weight) if all the codewords have the same weight. It is well known that constant weight or one weight codes have many important applications. In this paper, we study the structure of one weight $\ZZ$-linear and cyclic codes. We classify these type of one weight codes and also give some illustrative examples.
\end{abstract}
\begin{quotation}
\end{quotation}

\section{Introduction}

In algebraic coding theory, the most important class of codes is the family of linear codes. A linear code of length $n$ is a subspace $\C$ of a vector space
$F_{q}^n$ where $F_{q}$ is a finite field of size $q$. When $q=2$ then we have linear codes over $F_{2}$ which are called binary codes. Binary linear codes have very special and important place all among the finite field codes because of their easy implementations and applications. Beginning with a remarkable paper by Hammons et al. \cite{Hammons}, interest of codes over variety of rings have been increased. Such studies motivate the researchers to work on a different rings even over other structural algebras such as groups or modules. A $\Z_{4}$-submodule of $\Z_{4}^n$ is called a quaternary code. The structure of binary linear codes and quaternary linear codes have been studied in details for the last six decades. In 2010, Borges et al. introduced a new class of error
correcting codes over the ring $\mathbb{Z}_{2}^{\alpha}\times\mathbb{Z}_{4}^{\beta}$ called additive codes that generalizes the class of binary linear codes and the
class of quaternary linear codes  in \cite{3}. A $\mathbb{Z}_{2}\mathbb{Z}_{4}$-additive code $\mathcal{C}$ is defined to be a subgroup of $\mathbb{Z}_{2}^{\alpha}\times\mathbb{Z}_{4}^{\beta }$ where $\alpha+2\beta =n$. If $\beta =0$ then $\mathbb{Z}_{2}\mathbb{Z}_{4}$-additive codes are just binary linear codes, and if $\alpha=0,$ then $\mathbb{Z}_{2}\mathbb{Z}_{4}$-additive codes are the quaternary linear codes over $\mathbb{Z}_{4}$. $\Z_{2}\Z_{4}$-additive codes have been generalized to $\Z_{2}\Z_{2^s}$-additive codes in 2013 by Aydogdu and Siap in \cite{amis}, and recently this generalization has extended to $\Z_{p^r}\Z_{p^s}$-additive codes, for a prime $p$, by the same authors in \cite{zp}. Later, cyclic codes over $\mathbb{Z}_{2}^{\alpha}\times\mathbb{Z}_{4}^{\beta }$ have been introduced in \cite{7} in 2014 and more recently, in \cite{steven}, one weight codes over such a mixed alphabet have been studied. A code $\C$ is said to be one weight code if all the nonzero codewords of a code have the same Hamming weight where the Hamming weight of an any codeword is defined as the number of its nonzero coordinates. In \cite{Carlet}, Carlet determined one weight linear codes over $\Z_{4}$ and in \cite{Wood}, Wood studied linear one weight codes over $\Z_{m}$.
Constant weight codes are very useful in a variety of applications such as data storage, fault-tolerant circuit design and computing, pattern generation for circuit testing, identification coding, and optical overlay networks \cite{new}. Moreover, the reader can find the other applications of constant weight codes; determining the zero error decision feedback capacity of discrete memoryless channels in \cite{app1}, multiple access communications and spherical codes for modulation in \cite{app2,app3}, DNA codes in \cite{app4,app5}, powerline communications and frequency hopping in \cite{app7}.

Another important ring of four elements other than the ring $\Z_{4}$, is the ring $\Z_{2}+u\Z_{2}=R=\left\{
0,1,u,1+u\right\} $ where $u^{2}=0.$ For some of the works done in this direction we refer the reader to \cite{4,6,Dinh}. It has been shown that linear and cyclic codes over this ring have advantages compared to the ring $\mathbb{Z}_{4}.$ For an example; the finite field $GF(2)$ is a subring of the ring $R.$ So factorization over $GF(2)$ is still valid over
the ring $R$. The Gray image of any linear codes over $R$ is always a binary
linear codes which is not always the case for $\mathbb{Z}_{4}$.

In this work, we are interested in studying one weight codes over $\mathbb{Z}_{2}^{r}\times\left(\mathbb{Z}_{2}+u\mathbb{Z}_{2}\right)^{s}=\Z_{2}^{r}\times R^{s}$. This family of codes are special subsets of $\Z_{2}^r\times R^{s}$ which their all nonzero codewords have the same weight. We investigate and classify both linear and cyclic codes over $\Z_{2}^r\times R^{s}$ and also we give some one weight linear and cyclic code examples. Furthermore, we look at the Gray (binary) images of one weight cyclic codes over $\Z_{2}^{r}\times R^{s}$ and get optimal parameter binary codes.

\section{Preliminaries}

Let $R=\mathbb{Z}_{2}+u\mathbb{Z}_{2}=\left\{ 0,1,u,1+u\right\} $ be the
four-element ring with $u^{2}=0.$ It is easily seen that the ring $\mathbb{Z}_{2}$ is a subring of the ring $R.$ Then let us define the set
\begin{equation*}
\mathbb{Z}_{2}\mathbb{Z}_{2}[u]=\left\{ \left( a,b\right) |~a\in \mathbb{Z}_{2}\text{ and }b\in R \right\} .
\end{equation*}%
But we have a problem here, because the set $\mathbb{Z}_{2}\mathbb{Z}_{2}[u]$ is not well-defined with
respect to the usual multiplication by $u\in R$. So, we must define a new
method of multiplication on $\mathbb{Z}_{2}\mathbb{Z}_{2}[u]$ to make this set
as an $R$-module. Now define the mapping
\begin{eqnarray*}
\eta &:&R\rightarrow \mathbb{Z}_{2} \\
\eta \left( p+uq\right) &=&p.
\end{eqnarray*}
which means; $\eta (0)=0,~\eta (1)=1,~\eta (u)=0~$\ and $\eta (1+u)=1.$ It can be easily shown that $\eta $ is a ring homomorphism. Furthermore, for any element $d\in R,$ we can also define a scalar multiplication on $\mathbb{Z}_{2}\mathbb{Z}_{2}[u]$ as follows.
\begin{eqnarray*}
d\left( a,b\right) =\left( \eta (d)a,db\right) .
\end{eqnarray*}
This multiplication can be extended to the $\mathbb{Z}_{2}^{r}\times R^{s}$ for $d\in R$ and $v=(a_{0},a_{1},...,a_{r-1,}b_{0},b_{1},...,b_{s-1})\in\mathbb{Z}_{2}^{r}\times R^{s}$ as,
\begin{equation*}
dv=\left( \eta (d)a_{0},\eta (d)a_{1},...,\eta(d)a_{r-1,}db_{0},db_{1},...,db_{s-1}\right) .
\end{equation*}

\begin{lemma}
$\mathbb{Z}_{2}^{r }\times R^{s}$ is an $R-$module under
the multiplication defined above.
\end{lemma}

\begin{definition}
A non-empty subset $\mathcal{C}$ of $\mathbb{Z}_{2}^{r}\times R^{s}$ is called a $\mathbb{Z}_{2}\mathbb{Z}_{2}[u]$-linear
code if it is an $R$-submodule of $\mathbb{Z}_{2}^{r}\times R^{s}$.
\end{definition}
Now, take any element $a\in R,$ then there exist unique $p,q\in \mathbb{Z}_{2}$ such
that $a=p+uq$. Also note that the ring $R$ is isomorphic to $\mathbb{Z}_{2}^{2}$ as an additive group. Therefore, any $\mathbb{Z}_{2}\mathbb{Z}_{2}[u]-$linear code $\mathcal{C}$ is isomorphic to an
abelian group of the form $\mathbb{Z}_{2}^{k_{0}+k_{2}}\times \mathbb{Z}_{2}^{2k_{1}}$, where $k_{0}$, $k_{2}$ and $k_{1}$ are positive integers. Now
define the following sets.
\begin{equation*}
\mathcal{C}_{s}^{F}=\langle\{(a,b)\in \mathbb{Z}_{2}^{r}\times R^{s}~|~b~\text{free over}~R^{s}\}\rangle
\end{equation*}
where if $\langle b \rangle=R^{s}$ then $b$ is called free over $R^{s}$.
\begin{eqnarray*}
\mathcal{C}_{0} &=&\langle\{(a,ub)\in \mathbb{Z}_{2}^{r}\times R^{s}~|~a\neq 0\} \rangle\subseteq \mathcal{C}\backslash \mathcal{C}_{s}^{F} \\
\mathcal{C}_{1} &=&\langle\{(a,ub)\in \mathbb{Z}_{2}^{r}\times R^{s}~|~a=0\} \rangle \subseteq \mathcal{C}\backslash \mathcal{C}_{s}^{F}.
\end{eqnarray*}
Therefore, denote the dimension of $\mathcal{C}_{0},~\mathcal{C}_{1}$ and $\mathcal{C}_{s}^{F}$ as $k_{0},~k_{2}$ and $k_{1}$ respectively. Under these parameters, we say that
such a $\mathbb{Z}_{2}\mathbb{Z}_{2}[u]$-linear code $\mathcal{C}$ is of type $\left( r,s;k_{0},k_{1},k_{2}\right) $.
$\mathbb{Z}_{2}\mathbb{Z}_{2}[u]$-linear codes can be considered as binary
codes under a special Gray map. For $(x,y)\in \mathbb{Z}_{2}^{r}\times R^{s}$, where $(x,y)=\left( x_{0},x_{1},\ldots ,x_{r-1},y_{0},y_{1},\ldots
,y_{s-1}\right) $ and $y_{i}=p_{i}+uq_{i}$ the Gray map define as follows.
\begin{equation}\label{graymap}
\begin{split}
& \qquad \qquad \qquad \qquad \Phi :\mathbb{Z}_{2}^{r}\times
R^{s}\rightarrow \mathbb{Z}_{2}^{n}\text{ } \\
& \Phi \left( x_{0},\ldots x_{r-1},p_{0}+uq_{0},\ldots
p_{s-1}+uq_{s-1}\right) \\
& \qquad \qquad =\left( x_{0},\ldots x_{r-1},q_{0},\ldots
,q_{s-1},p_{0}\oplus q_{0},\ldots ,p_{s-1}\oplus q_{s-1}\right),
\end{split}
\end{equation}
where $p_{i}\oplus q_{i}=p_{i}+q_{i}$ and $n=r+2s$.

The minimum distance of a linear code $\C$, denoted by $d(\C)$ is defined by
$$
d(\C)=min\{d(c_{1},c_{2}):c_{1},~c_{2}\in \C, c_{1}\neq c_{2}\}.
$$

The Hamming weight of a codeword $c$, denoted by $wt_{H}(c)$, is the number of its nonzero coordinates, i.e., $wt_{H}(c)=d(c,0)$ where $0$ is the zero codeword.
The Hamming distance between two codewords is the Hamming weight of their difference and the Lee distance for the
codes over $R$ is the Lee weight of their differences where the Lee weights of the elements of $R$  are defined as $wt_{L}(0)=0,~wt_{L}(1)=1,~wt_{L}(u)=2$ and $wt_{L}(1+u)=1$. The Gray map defined above is a distance preserving map which transforms the Lee distance in $\mathbb{Z}_{2}^{r}\times R^{s}$ to the Hamming distance in $\mathbb{Z}_{2}^{n}$.
Furthermore, for any $\mathbb{Z}_{2}\mathbb{Z}_{2}[u]$-linear code $\mathcal{C},$ we have that $\Phi \left(\mathcal{C}\right) $ is a binary linear code as well. This property is not valid for the $\mathbb{Z}_{2}\mathbb{Z}_{4}$-additive codes. And also, we always have
\begin{eqnarray*}
wt(v)=wt_{H}(v_{1})+wt_{L}(v_{2}),
\end{eqnarray*}
where $wt_{H}(v_{1})$ is the Hamming of weight of $v_{1}$ and $wt_{L}(v_{2})$ is the Lee weight of $v_{2}.$
If $\mathcal{C}$ is a $\mathbb{Z}_{2}\mathbb{Z}_{2}[u]$-linear code of type $\left( r,s;k_{0},k_{1},k_{2}\right) $
then the binary image $C=\Phi (\mathcal{C})$ is a binary linear code of
length $n=r+2s$ and size $2^{n}$. It is also called a $\mathbb{Z}_{2}\mathbb{Z}_{2}[u]$-linear code.
Now, let
\begin{equation*}
v=\left( a_{0},\ldots ,a_{r-1},b_{0},\ldots ,b_{s-1}\right) ,w=\left(
d_{0},\ldots ,d_{r-1},e_{0},\ldots ,e_{s-1}\right) \in \mathbb{Z}%
_{2}^{r}\times R^{s}
\end{equation*}
be any two elements. Then we can define the inner product as
\begin{equation*}
\left\langle v,w\right\rangle =\left(
u\sum_{i=0}^{r-1}a_{i}d_{i}+\sum_{j=0}^{s-1}b_{j}e_{j}\right) \in \mathbb{Z}%
_{2}+u\mathbb{Z}_{2}.
\end{equation*}%
According to this inner product, the dual linear code $\mathcal{C}^{\perp }$
of a $\mathbb{Z}_{2}\mathbb{Z}_{2}[u]$-linear code $\mathcal{C}$ is
also defined in a usual way easily,
\begin{equation*}
\mathcal{C}^{\perp }=\left\{ w\in \mathbb{Z}_{2}^{r}\times
R^{s}|~\left\langle v,w\right\rangle =0~\text{for all}~v\in \mathcal{C}%
\right\} .
\end{equation*}
Hence, if $\mathcal{C}$ is a $\mathbb{Z}_{2}\mathbb{Z}_{2}[u]$-linear
code, then $\mathcal{C}^{\perp }$ is also a $\mathbb{Z}_{2}\mathbb{Z}_{2}[u]$-linear code.

The standard forms of generator and parity-check matrices of a $
\mathbb{Z}_{2}\mathbb{Z}_{2}[u]$-linear code $\mathcal{C}$ is given as follows.

\begin{theorem}\label{Matrices}\cite{2}
Let ${\mathcal{C}}$ be a ${\mathbb{Z}_{2}\mathbb{Z}_{2}[u]}$-linear code of
type ${(r,s;k_{0},k_{1},k_{2})}$. Then the standard forms of the generator and the parity-check
matrices of ${\mathcal{C}}$ are:

\begin{equation*}
G= \left[
\begin{array}{cc|ccc}
I_{k_{0}} & A_{1} & 0 & 0 & uT \\ \hline
0 & S & I_{k_{1}} & A & B_{1}+uB_{2} \\
0 & 0 & 0 & uI_{k_{2}} & uD%
\end{array}%
\right]
\end{equation*}

\begin{equation*}
H= \left[
\begin{array}{cc|ccc}
-A_{1}^{t} & I_{r-k_{0}} & -uS^{t} & 0 & 0 \\
-T^{t} & 0 & -(B_{1}+uB_{2})^{t}+D^{t}A^{t} & -D^{t} & I_{s-k_{1}-k_{2}} \\
0 & 0 & -uA^{t} & uI_{k_{2}} & 0%
\end{array}
\right]
\end{equation*}
where $A,~A_{1},~B_{1},~B_{2},$ $D,~S$ and $T$ are matrices over $\mathbb{Z}
_2$.
\end{theorem}

Therefore, we can conclude the following corollary.
\begin{corollary}
If ${\mathcal{C}}$ is a ${\mathbb{Z}_{2}\mathbb{Z}_{2}[u]}$-linear code of
type ${(r,s;k_{0},k_{1},k_{2})}$ then ${\mathcal{C}}^{\perp }$ is of type $%
\left( r,s;r-k_{0},s-k_{1}-k_{2},k_{2}\right) $.
\end{corollary}

The weight enumerator of an any ${\mathbb{Z}_{2}\mathbb{Z}_{2}[u]}$-linear code $\C$ of
type ${(r,s;k_{0},k_{1},k_{2})}$ is defined as
\begin{equation*}
W_{{\mathcal{C}}}(x,y)=\sum\limits_{c\in {\mathcal{C}}}x^{n-w\left( c\right)
}y^{w\left( c\right) }.
\end{equation*}
where, $n=r+2s.$
Moreover, the MacWilliams relations for codes over $\ZZ$ can be given as follows.
\begin{theorem}\label{Macw}\cite{2}
Let ${\mathcal{C}}$ be a ${\mathbb{Z}_{2}\mathbb{Z}_{2}[u]}-$linear code.
The relation between the weight enumerators of ${\mathcal{C}}$ and its dual
is
\begin{equation*}
W_{{\mathcal{C}}^{\perp}}\left(x,y\right)=\frac{1}{{\left|{\mathcal{C}}%
\right| }}W_{{\mathcal{C}}}\left(x+y,x-y\right).
\end{equation*}
\end{theorem}
\bigskip
We have given some information about the general concept of codes over $\mathbb{Z}_{2}^{r}\times\left(\mathbb{Z}_{2}+u\mathbb{Z}_{2}\right)^{s}$. To make reader understanding the paper easily we give the following example.

\begin{example}
Let $\C$ be a linear code over $\Z_{2}^{3}\times \left(\Z_{2}+u\Z_{2}\right)^{4}$ with the following generator matrix.
\begin{equation*}
\left[
\begin{array}{ccc|cccc}
1 & 1 & 0 & 0 & u  &  u & u \\
0 & 1 & 1 & 1 & 1+u &  u & 0 \\
0 & 1 & 0 & u & u  & u  & 0
\end{array}
\right].
\end{equation*}

We will find the standard form of the generator matrix of $\C$ and then using this standard form, we find the generator matrix of the linear dual code $\C^{\perp}$ and also we determine the types of both $\C$ and its dual.

Now, applying elementary row operations to above generator matrix, we have the standard form as follows.

\begin{equation*}
G=\left[
\begin{array}{ccc|cccc}
1 & 0 & 0 & 0 & u  &  0 & u \\
0 & 1 & 0 & 0 & 0 &  u & 0 \\ \hline
0 & 0 & 1 & 1 & 1+u  & 0  & 0
\end{array}
\right].
\end{equation*}
Since, $G$ is in the standard form we can write this matrix as

\begin{equation*}
G=
\left[
\begin{array}{ccc||cccc}
\cellcolor{yellow!50}1& \cellcolor{yellow!50}0& \cellcolor{blue!50}0& \cellcolor{gray!50}0& \cellcolor{green!50}u& \cellcolor{green!50}0 & \cellcolor{green!50}u \\
\cellcolor{yellow!50}0 & \cellcolor{yellow!50}1& \cellcolor{blue!50}0 & \cellcolor{gray!50}0& \cellcolor{green!50}0& \cellcolor{green!50}u & \cellcolor{green!50}0\\ \hline\hline
\cellcolor{gray!50}0 & \cellcolor{gray!50}0 & \cellcolor{cyan!50}1 & \cellcolor{yellow!50}1 & \cellcolor{red!50}1+u & \cellcolor{red!50}0 & \cellcolor{red!50}0
\end{array}\right]
=
\left[
\begin{array}{cc||cc}
\cellcolor{yellow!50}I_{k_0}& \cellcolor{blue!50}A_{1}& \cellcolor{gray!50}0& \cellcolor{green!50}uT \\ \hline\hline
\cellcolor{gray!50}0 & \cellcolor{cyan!50}S_{1}& \cellcolor{yellow!50}I_{k_1} & \cellcolor{red!50}B_{1}+uB_{2}
\end{array}\right].
\end{equation*}
Hence, with the help of Theorem \ref{Matrices} the parity-check matrix of $\C$ is
\begin{equation*}
H=\left[\begin{array}{ccc|cccc}
0 & 0 & 1 & u &  0 & 0&0 \\
1 & 0 & 0 & 1+u& 1 & 0&0 \\
0 & 1 & 0 & 0 & 0 & 1&0\\
1 & 0 & 0 & 0 & 0 & 0&1
\end{array}\right].
\end{equation*}

Therefore,

\begin{itemize}
\item
$\C$ is of type $(3,4;2,1,0)$ and has $2^{2}4^{1}=16$ codewords.

\item

$\C^{\perp}$ is of type $(3,4;1,3,0)$ and has $2^{1}4^{3}=128$ codewords.

\item
$\C=\{(0,0,0,|0,0,0,0),(1,0,0,|0,u,0,u),(0,1,0,|0,0,u,0),(0,0,1,|1,1+u,0,0),(0,0,0,|u,u,0,0),
(0,0,1,|1+u,1,0,0),(1,1,0,|0,u,u,u),(1,0,1,|1,1,0,u),(0,1,1,|1,1+u,u,0),(1,1,1,|1,1,u,u),(1,0,0,|u,0,0,u),
(0,1,0,|u,u,u,0),(1,1,0,|u,0,u,u),(1,0,1,|1+u,1+u,0,u),(0,1,1,|1+u,1,u,0),(1,1,1,|1+u,1+u,u,u)\}.$
\item

$W_{\C}(x,y)=x^{11}+3 x^8 y^3+x^7 y^4+2 x^6 y^5+4 x^5 y^6+x^4 y^7+2 x^3 y^8+2 x^2 y^9$.

\item
$W_{\C^{\perp}}(x,y)=\frac{1}{|\C|}W_{\C}(x+y,x-y)=x^{11}+6 x^9 y^2+8 x^8 y^3+16 x^7 y^4+32 x^6 y^5+26 x^5 y^6+24 x^4 y^7+15 x^3 y^8$.

\item
The Gray image $\Phi(\C)$ of $\C$ is a $[11,4,3]$ binary linear code.

\item

$\Phi(\C^{\perp})$ is a $[11,7,2]$ binary linear code.

\end{itemize}
\end{example}

\section{The Structure of One Weight $\ZZ$-linear Codes}

In this part of the paper, we study the structure of one weight codes over $\Z_{2}^{r}\times R^{s}$.

\begin{definition}
Let $\C$ be a $\ZZ$-linear code. $\C$ is called a one (constant) weight code if all of its nonzero codewords have the same weight. Furthermore, if the weight of $\C$ is $\m$ then it is called a code with weight $\m$.
\end{definition}

\begin{definition}
Let $c_{1},c_{2},e_{1},e_{2}$ be any four distinct codewords of $\ZZ$-linear code $\C$. If the distance between $c_{1}$ and $e_{1}$ is equal to the distance between $c_{2}$ and $e_{2}$, that is, $d(c_{1},e_{1})=d(c_{2},e_{2})$, then $\C$ is said to be equidistant.
\end{definition}

\begin{theorem}
Let $\C$ be an equidistant $\ZZ$-linear code with distance $\m$. Then $\C$ is a one weight code with weight $\m$. Moreover, the binary image $\Phi(\C)$ of $\C$ is also a one weight code with weight $\m$.
\end{theorem}

\begin{proof}
Let $0\neq c\in \C$. Since $\C$ is a $\ZZ$-linear code, we have $wt(c)=d(c,0)=\m$. And also, for the binary image $\Phi(\C)$, $wt_{H}(\Phi(\C))=d_{H}\left(\Phi(\C),0\right)=d(c,0)=wt(c)=\m.$
Here, we can also say that a one weight $\ZZ$-linear code $\C$ with weight $\m$ is also an equidistant code with distance $\m$.
\end{proof}

\begin{example}
It is worth to note that the dual of a one weight code is not necessarily a one weight code. Also it may be interesting to give in this example a codeword of weight 2.
Let $\C$ be a $\ZZ$-linear code of type $(2,2;0,1,0)$ with $\C=\langle (1,1|1+u,1+u)\rangle$. Then $\C=\{(0,0|0,0),(1,1|1+u,1+u),(1,1|1,1),(0,0|u,u)\}$ and $\C$ is a one weight code with weight $\m=4$. On the other hand, the dual code $\C^{\perp}$ is generated by $\langle(1,0|u,0),(0,1|u,0),(0,0|1,1)\rangle$ and of type $(2,2;2,1,0)$. But $d(\C^{\perp})=2$ and $\C^{\perp}$ is not a one weight code.
\end{example}

\begin{example}
Let $\C$ be a $\ZZ$-linear code with the standard form of the generator matrix $\left[\begin{array}{ccc|cc}
  1 & 0 & 1 & 0 & u \\ \hline
  0 & 1 & 1 & 1 & 1+u
\end{array}\right]$, then $\C$ is of type $(3,2;1,1,0)$ and one weight code with weight $4$. Furthermore, $\Phi(\C)$ is a binary linear code with parameters $[7,3,4]$. Here, note that the binary image of $\C$ is the binary simplex code of length $7$, which is the dual of the $[7,4,3]$ Hamming code.
\end{example}

Now, we give a theorem which gives a construction of one weight codes over $\Z_{2}^{r}\times R^{s}$.
\begin{theorem}

Let $\C$ be a one weight $\ZZ$-linear code of type $(r,s;k_{0},k_{1},k_{2})$ and weight $\m$. Then, one weight code of type $(\gamma r,\gamma s;k_{0},k_{1},k_{2})$ with weight $\gamma\m$ exists, where $\gamma$ is a positive integer.
\end{theorem}

\begin{proof}
Let $G=[G_{r}|G_{s}]$ be a generator matrix of $\ZZ$-linear code $\C$ of type $(r,s;k_{0},k_{1},k_{2})$ with weight $\m$, where $G_{r}$ denotes the first $r$ part of $G$ and $G_{s}$ denotes the last $s$ part of $G$. We can copy the first $r$ part and the last $s$ part of the generator matrix $G$ as
$\bar{G}=[\underbrace{G_{r}\cdots G_{r}}_\text{$\gamma$ times}| \underbrace{G_{s}\cdots G_{s}}_\text{$\gamma$ times}]$. Copying the rows of $G$ does not change the type but the length of the new code is $\gamma(r+s)$. Since the weight of any codeword in $\C$ is $\m$ then the new code has the weight $\gamma\m$.
\end{proof}

\begin{definition}
Let ${\mathcal{C}}$ be a ${\mathbb{Z}_{2}\mathbb{Z}_{2}[u]}$-linear code. Let $%
{\mathcal{C}}_{r}$ (respectively ${\mathcal{C}}_{s}$) be the punctured code
of ${\mathcal{C}}$ by deleting the coordinates outside $r$ (respectively $s$). If ${\mathcal{C}}={\mathcal{C}}_{r}\times {\mathcal{C}}_{s}$ then ${%
\mathcal{C}}$ is called separable.
\end{definition}

\begin{corollary}\label{corol1}
There do not exist separable one weight $\ZZ$-linear codes.
\end{corollary}

\begin{proof}
Let ${\mathcal{C}}={\mathcal{C}}_{r}\times {\mathcal{C}}_{s}$ be a separable one weight code over $\Z_{2}^{r}\times R^s$. Consider the codeword $0\neq(v,w)\in \Z_{2}^{r}\times R^s$ with weight $\m$. Since $\C$ is a linear code then we have both $0\in \C_{r}$ and $0\in \C_{s}$. Hence, $(v,0)$ and $(0,w)$ are elements of $\Z_{2}^{r}\times R^s$. If $v\neq0$ and $w\neq0$ then $wt(v,w)\neq wt(v,0)$ and $wt(v,w)\neq wt(0,w)$. So, we have a contradiction. As a result, there is no separable one weight $\ZZ$-linear code.
\end{proof}

\begin{lemma}
If $\C$ is a $\ZZ$-linear code of type $\ty$ with no all zero columns in the generator matrix of $\C$. Then the sum of the weights of all codewords of $\C$ is equal to $\frac{|\C|}{2}(r+2s)$.
\end{lemma}

\begin{proof}
Let $G$ be a matrix whose rows are all codewords of $\C$. Since $\C$ is a linear code, in the first $r$ columns, the number of coordinates containing $0$ is equal to the number of coordinates containing $1$. And for the last $s$ columns, the column contains either the same number of $0,1,u,1+u$ or the same number of $0$ and $u$. Assume that there are $x$ columns containing only $0$ and $u$ in the last $s$ columns. So, it remains $s-x$ columns containing $0,1,u,1+u$. Therefore, the sum of the all codewords of $\C$ is
\begin{eqnarray*}
\sum_{c\in \C}wt(c)&=&r\frac{|\C|}{2}+2x\frac{|\C|}{2}+(s-x)\left(\frac{|\C|}{4}+2\frac{|\C|}{4}+\frac{|\C|}{4}\right)\\
&=& r\frac{|\C|}{2}+s|\C|=\frac{|\C|}{2}(r+2s).
\end{eqnarray*}
\end{proof}

\begin{theorem}\label{weight}
Let $\C$ be a one weight $\ZZ$-linear code of type $\ty$ such that there are no zero columns in the generator matrix of $\C$. Then, if the weight is $\m$, $\m=\alpha~2^{(k_{0}+2k_{1}+k_{2})-1}$ where $\alpha$ is a positive integer satisfying $(r+2s)=\alpha\left(2^{k_{0}+2k_{1}+k_{2}}-1 \right)$. In addition, if $\m$ is an odd integer, then $r$ is also odd and $\C=\langle\underbrace{1\cdots 1}_\text{$r$ times}| \underbrace{u\cdots u}_\text{$s$ times}\rangle$.
\end{theorem}

\begin{proof}
Since $\C$ is a one weight code of weight $\m$ then the sum of the weights of all codewords of $\C$ is $\left(|\C|-1\right)\m$. And also we know from the above lemma that $\sum_{c\in\C}wt(c)=\frac{|\C|}{2}(r+2s)$, so we have
 $$\frac{|\C|}{2}(r+2s)=\left(|\C|-1\right)\m.$$
Further, since $\C$ is of type $\ty$ then it is isomorphic to an abelian group $\Z_{2}^{k_{0}}\times \Z_{2}^{2k_{1}}\times \Z_{2}^{k_{2}}$ then $|\C|=2^{k_{0}+2k_{1}+k_{2}}$. Also, note that $gcd\left(|\C|-1,\frac{|\C|}{2}\right)=1$. Hence, there is a positive integer $\alpha$, such that $\m=\alpha~\frac{|\C|}{2}=\alpha~2^{(k_{0}+2k_{1}+k_{2})-1}$ and $(r+2s)=\alpha~\left(|\C|-1\right)=\alpha~\left(2^{k_{0}+2k_{1}+k_{2}}-1\right)$.
Further, in the case where $\m$ is odd, then $\alpha~2^{k_{0}+2k_{1}+k_{2}-1}$ is odd. So, $\alpha$ is odd and $\alpha~2^{k_{0}+2k_{1}+k_{2}-1}=1$. Therefore, $k_{1}=0$ and $k_{0}+k_{2}=1$. Since $\m$ is odd then $k_{2}$ must be also zero. Because, the subcode generated by the rows of dimension $k_{2}$ only consists of $u$'s. So, if $k_{2}\neq 0$ then $\m$ is even and this is contradiction. Finally, $\alpha=\m=(r+2s)$ and again since $\m$ is odd then $r$ is also odd and only codeword that satisfies these conditions is $(\underbrace{1\cdots 1}_\text{$r$ times}| \underbrace{u\cdots u}_\text{$s$ times})$.
\end{proof}

The following theorem gives a complete classification of dual one weight $\ZZ$-linear codes.

\begin{theorem}
Let $\C$ be a one weight $\ZZ$-linear code of type $\ty$ and weight $\m$. If there is no zero columns in the generator matrix of $\C$, then $d(\C^{\perp})\geq2.$ Also, $d(\C^{\perp})\geq3$ if and only if $\alpha=1$.
\end{theorem}

\begin{proof}
Since there is no zero columns in the generator matrix of $\C$ then from Theorem \ref{weight}, there is a positive integer $\alpha$ such that $\m=\alpha~2^{(k_{0}+2k_{1}+k_{2})-1}$ and $n=(r+2s)=\alpha~\left(2^{k_{0}+2k_{1}+k_{2}}-1\right)$. We can write the weight enumerator of the one weight linear code $\C$ as $W_{\C}(x,y)=x^{n}+\left(|\C|-1\right)x^{n-\m}y^{\m}$. Therefore, by Theorem \ref{Macw}, the weight enumerator for the dual code $\C^{\perp}$ is
\begin{eqnarray}\label{eq1}
W_{\C^{\perp}}(x,y)=\frac{1}{|\C|}\left( (x+y)^{n}+(|\C|-1)(x+y)^{n-\m}(x-y)^{\m}\right).
\end{eqnarray}
Let us calculate the coefficient of $x^{n-1}y$ as
\[
\frac{1}{|\C|} \left( {n\choose 1}+(|\C|-1)\left({n-\m\choose 1}-{\m\choose 1}\right) \right)=\frac{1}{|\C|}\left(n|\C|-2\m(|\C|-1)\right).
\]
Since, $|\C|=2^{k_{0}+2k_{1}+k_{2}}$, then $n=(r+2s)=\alpha~\left(2^{k_{0}+2k_{1}+k_{2}}-1\right)=\alpha~(|\C|-1)$ and $2\m=\alpha~|\C|$. Therefore, we have
\[
\frac{1}{|\C|}\left(n|\C|-2\m(|\C|-1)\right)=0,
\]
which means there is no codeword with weight $1$ in $\C^{\perp}$, so $d(\C^{\perp})\geq2$.

Now, let the coefficient of $x^{n-2}y^{2}$ be $\lambda$. Considering $n=\alpha(|\C|-1)$ and $\m=\frac{\alpha|\C|}{2}$, we have
\begin{eqnarray*}
\lambda&=&\frac{1}{|\C|} \left( {n\choose 2}+(|\C|-1)\left({n-\m\choose 2}+{\m\choose 2}-{n-\m\choose 1}{\m\choose 1}\right) \right)\\
&=&\frac{1}{|\C|}\left( \frac{n(n-1)}{2}+\frac{(|\C|-1)(n^{2}-4n\m+4\m^{2}-n)}{2}\right)\\
&=&\frac{\alpha(\alpha-1)(|\C|-1)}{2}.
\end{eqnarray*}
Therefore, $\lambda=0$ if and only if $\alpha=1$.
\end{proof}

\begin{corollary}
For $\alpha=1$, if $|\C|\geq4$ then $d(\C^{\perp})=3.$
\end{corollary}

\begin{proof}
Let the coefficient of $x^{n-3}y^3$ in Equation \ref{eq1} be $\kappa$. Then,

\begin{eqnarray*}
\kappa&=&\frac{1}{|\C|}\left( {n\choose3}+(|\C|-1)\left({n-\m\choose3}-{n-\m\choose2}{\m\choose1}+{n-\m\choose1}{\m\choose2}-{\m\choose3} \right)                       \right)\\
&=&\frac{1}{|\C|}\left(\frac{n(n-1)(n-2)+6(|\C|-1)(\m-1)}{6} \right).
\end{eqnarray*}
Again, considering $n=|\C|-1$ and $\m=\frac{|\C|}{2}$ (since $\alpha=1$), we have

\[
\kappa=\frac{(|\C|-1)(|\C|-2)}{6}.
\]
Since, $|\C|\geq4$ then $\kappa\geq1$, that is, $d(\C^{\perp})=3$.
\end{proof}

We have proved that if $\C$ is a one weight $\ZZ$-linear code of type $\ty$ and weight $\m$ then there is a positive integer $\alpha$ such that $\m=\alpha~2^{(k_{0}+2k_{1}+k_{2})-1}$, so the minimum distance for a one weight $\ZZ$-linear code must be even. In the following, we characterize the structure of $\ZZ$-linear codes.

\begin{theorem}\label{item}
Let $\C$ be a one weight $\ZZ$-linear code over $\Z_{2}^{r}\times R^{s}$ with generator matrix $G$ and weight $\m$.

\begin{itemize}
  \item [i)]
If $v=(a|b)$ is an any row of $G$, where $a=(a_{0},\ldots ,a_{r-1})\in \Z_{2}^r$ and $b=(b_{0},\ldots ,b_{s-1})\in R^{s}$, then the number of units(1 or $1+u$) in $b$ is either zero or $\frac{\m}{2}$.

  \item [ii)]
If $v=(a|b)$ and $w=(c|d)$ are two distinct rows of $G$, where $b$ and $d$ are free over $R^{s}$, then the coordinate positions where $b$ has units ($1$ or $1+u$) are the same that the coordinate positions where $d$ has units.

  \item [iii)]
If $v=(a|b)$ and $w=(c|d)$ are two distinct rows of $G$, where $b$ and $d$ are free over $R^{s}$, then $|\{j:(b_{j})=(d_{j})=1~or~1+u\}|= |\{j:(b_{j})=1,(d_{j})=1+u~or~(b_{j})=1+u,(d_{j})=1\}|=\frac{\m}{4}$.

\end{itemize}
\end{theorem}

\begin{proof}
\begin{itemize}
\item[i)]
The weight of $v=(a|b)$ is $wt(v)=wt_{H}(a)+wt_{L}(b)=\m$. Since $\C$ is linear $uv=(0|ub)$ is also in $\C$ then, if $ub=0$ then $b$ does not contain units. If $ub\neq0$, then $wt(v)=\m=0+wt_{L}(ub)$ and therefore, $wt_{L}(ub)=2|\{j:(b)_{j}=1~\text{or}~1+u\}|=\m$. Hence, the number of units in $b$ is $\frac{\m}{2}$.

\item[ii)]

Multiplying $v$ and $w$ by $u$ we have, $uv=(0|ub)$ and $uw=(0|ud)$. If $v$ and $w$ have units in the same coordinate positions then we get $uv+uw=0$. So, assume that they have some units in different coordinates. Since $\C$ is a one weight code with weight $\m$, if $uv+uw\neq0$ then the number of coordinates where $b$ and $d$ have units in different places must be $\frac{\m}{2}$. To obtain this, the number of coordinates where $\{(b_{j})=1=(d_{j})\}$ and $\{(b_{j})=1+u=(d_{j})\}$ has to be $\frac{\m}{2}$, and in all other coordinates where $\{(b_{j})=1~or~1+u\}$ we need $\{(d_{j})=0~or~u \}$, and also in all other coordinates where $\{(b_{j})=0~or~u\}$ we need $\{(d_{j})=1~or~1+u\}$. Hence, consider the vector $v+(1+u)w$. This vector has the same weight as $v+w$ in the first $r$ coordinates but for the last $s$ coordinates, it has $u's$ in the coordinates where $\{(b_{j})=1=(d_{j})\}$ and $\{(b_{j})=1+u=(d_{j})\}$, so its weight is greater than $\m$. This contradiction gives the result.

\item[iii)]
Let $x=v+w$ and $y=v+(1+u)w$ be two vectors in $\C$. The binary parts of these two vectors are the same, and for the coordinates over $R^{s}$ we know from ii) that $v$ and $w$ have units in the same coordinate positions, and for the all other coordinates in $R^{s}$, the values of $x$ and $y$ are the same. Therefore, the sum of the weights of the units in $v$ must be same in $x$ and $y$. So, they also have the same number of coordinates with $u$. But this is only possible if $|\{j:(b_{j})=(d_{j})=1~or~1+u\}|= |\{j:(b_{j})=1,(d_{j})=1+u~or~(b_{j})=1+u,(d_{j})=1\}|$. We also know from i) that the number of units in $v$ is $\frac{\m}{2}$, so we have the result.

\end{itemize}
\end{proof}

\begin{theorem}\label{Teo20}
Let $\C$ be a one weight code of type $\ty$. Then $k_{1}\leq1$ and $\C$ has the following standard form of the generator matrices.
If $k_{1}=0$ then
\begin{equation*}
G= \left[
\begin{array}{cc|cc}
I_{k_{0}} & A_{1}  & 0 & uT \\ \hline
0 & 0 & uI_{k_{2}} & uD
\end{array}%
\right].
\end{equation*}

If $k_{1}=1$ then
\begin{equation*}
G= \left[
\begin{array}{cc|ccc}
I_{k_{0}} & A_{1} & 0 & 0 & uT \\ \hline
0 & s & 1 & a & b_{1}+ub_{2} \\
0 & 0 & 0 & uI_{k_{2}} & uD
\end{array}
\right]
\end{equation*}
where $s,a,b_{1},b_{2}$ are vectors over $\Z_{2}.$
\end{theorem}

\begin{proof}
From Theorem \ref{item} i), we know that any two distinct free vectors have their units in the same coordinate positions. So, if we add the first free row of the generator matrix to the other rows, we have only one free row in the generator matrix. Hence, $k_{1}\leq1$ and considering this and using the standard form of the generator matrix for a $\ZZ$-linear code $\C$ given in Theorem \ref{Matrices}, we have the result.
\end{proof}

\section{One Weight $\ZZ$-cyclic Codes}

In this section, we study the structure of one weight $\ZZ$-cyclic codes. At the beginning, we give some fundamental definitions and theorems about $\ZZ$-cyclic codes. This information about $\ZZ$-cyclic codes was given in \cite{ieee}, with details.

\begin{definition}
An $R$-submodule $\mathcal{C}$ of $\mathbb{Z}_{2}^{r}\times R^{s}$
is called a $\mathbb{Z}_{2}\mathbb{Z}_{2}[u]$-cyclic code
if for any codeword $v=\left( a_{0},a_{1},\ldots ,a_{r-1},b_{0},b_{1},\ldots, b_{s-1}\right) \in \mathcal{C},$ its
cyclic shift
\begin{equation*}
T(v)=\left( a_{r-1},a_{0},\ldots, a_{r-2},b_{s-1},b_{0},\ldots, b_{s-2}\right)
\end{equation*}
is also in $\mathcal{C}$.
\end{definition}

Any codeword $c=\left( a_{0},a_{1},\ldots, a_{r-1},b_{0},b_{1},\ldots, b_{s-1}\right) \in\mathbb{Z}_{2}^{r}\times
R^{s}$ can be identified with a module element such that
\begin{eqnarray*}
c(x) &=&\left( a_{0}+a_{1}x+\ldots +a_{r-1}x^{r-1},b_{0}+b_{1}x+\ldots +b_{s-1}x^{s-1}\right) \\
&=&\left( a(x),b(x)\right)
\end{eqnarray*}
in $R_{r,s}=\mathbb{Z}_{2}[x]/\left( x^{r}-1\right) \times
R[x]/\left( x^{s}-1\right).$ This identification gives a one-to-one correspondence
between elements in $\mathbb{Z}_{2}^{r}\times R^{s}$ and elements
in $R_{r,s}.$

\begin{theorem}\cite{ieee}
\label{Main}Let $\mathcal{C}$ be a $\ZZ$-cyclic code in $R_{r,s }.$ Then we can identify $
\mathcal{C}$ uniquely as $\mathcal{C}=\langle \left(
f(x),0\right) ,\left( l(x),g(x)+ua(x)\right) \rangle ,$ where $f(x)|\left(
x^{r}-1\right)(\bmod~2),$ and $a(x)|g(x)|\left( x^{s}-1\right) \left(\bmod~2\right),$ and $l(x)$ is
a binary polynomial satisfying $\deg (l(x))<\deg (f(x)),$ $f(x)|\left(
\dfrac{x^{s}-1}{a(x)}\right) l(x)\left(\bmod~2 \right)$ and $
f(x)\neq\left( \dfrac{x^{s}-1}{a(x)}\right)l(x)\left(\bmod~2\right)$.
\end{theorem}

Considering the theorem above, the type of $\mathcal{C}=\langle \left(
f(x),0\right) ,\left( l(x),g(x)+ua(x)\right) \rangle$ can be written in terms of the degrees of the polynomials $f(x),~a(x)$ and $g(x)$. Let $t_{1}=\deg f(x),~t_{2}=\deg g(x)$ and $t_{3}=\deg a(x)$. Then ${\mathcal{C}}$ is
of type (\cite{ieee})
\begin{eqnarray*}
\left( r,s;r-t_{4},s-t_{2},t_{2}+t_{4}-t_{1}-t_{3}\right)
\end{eqnarray*}
where $d_{1}(x)=\gcd \left( f(x),\dfrac{x^{s}-1}{g(x)}l(x)\right) $ and $t_{4}=\deg d_{1}(x)$.

\begin{corollary}\label{Cor}
If $\C$ is a one weight cyclic code generated by $\left(l(x),g(x)+ua(x)\right)\in R_{r,s}$ with weight $\m$ then $\m=2s$.
\end{corollary}
 \begin{proof}
 We know from Theorem \ref{Teo20} that if $\C$ is a one weight $\ZZ$-linear code then $k_{1}$, which generates the free part of the code, is less than or equal to $1$. So, in the case where $\C$ is cyclic, it means that $s-t_{2}\leq1$, where $t_{2}=\deg g(x)$. Therefore we have $\deg g(x)=s-1$ and the polynomial $g(x)+ua(x)$ generates the vector with all unit entries and length $s$. If we multiply the whole vector (length$=r+s$) by $u$, then we have a vector with all entries $0$ in the first $r$ coordinates and all coordinates $u$ in the last $s$ coordinates. So the weight of this vector is $2s.$ Hence the weight of $\C$ must be $2s$.
\end{proof}
\begin{theorem}\cite{ieee}
\label{Theorem15}Let $\mathcal{C}=\langle \left( f(x),0\right) ,\left(
l(x),g(x)+ua(x)\right) \rangle $ be a cyclic code in $R_{r,s}$
where $f(x),~l(x),~g(x)$ and $a(x)$ are as in Theorem \ref{Main} and $
f(x)h_{f}(x)=x^{r}-1$, $g(x)h_{g}(x)=x^{s}-1$, $g(x)=a(x)b(x)$.

Let
\begin{equation*}
S_1=\bigcup_{i=0}^{deg(h_{f})-1}\left\{x^i \ast \left( f(x),0\right)\right\},
\end{equation*}

\begin{equation*}
S_{2}=\bigcup_{i=0}^{deg(h_{g})-1}\left\{ x^{i}\ast \left(
l(x),g(x)+ua(x)\right) \right\}
\end{equation*}

and
\begin{equation*}
S_{3}=\bigcup_{i=0}^{deg(b)-1}\left\{ x^{i}\ast \left(h_g(x)l(x),uh_{g}(x)a(x)\right) \right\} .
\end{equation*}

Then $S=S_{1}\cup S_{2}\cup S_{3}$ forms a minimal spanning set for $\mathcal{C}$ as a $R$-module.
\end{theorem}

Let $\C=\langle\left(f(x),0\right),\left(l(x),g(x)+ua(x)\right)\rangle$ be a one weight cyclic code in $R_{r,s}$. Consider the codewords $(v,0)\in \langle\left(f(x),0\right)\rangle$ and $(w_{r},w_{s})\in \langle \left(l(x),g(x)+ua(x)\right)\rangle$. Since $\C$ is a one weight code, $wt(v,0)=wt(w_{r},w_{s})$. Further, since $\C$ is a $R$-submodule, $u(w_{r},w_{s})=(0,uw_{s})\in\C$ and $wt(v,0)=wt(0,uw_{s})$. Moreover, $(v,uw_{s})\in \C$ because of the linearity of $\C$. But it is clear that $wt(v,uw_{s})\neq wt(v,0)$ and $wt(v,uw_{s})\neq wt(0,uw_{s})$. Hence, $\langle(f(x),0)\rangle$ can not generate a one weight code.

\bigskip

Now, let us suppose that $\C=\langle \left(l(x),g(x)+ua(x)\right)\rangle$ is a one weight cyclic code in $R_{r,s}$. We know from Corollary \ref{Cor} that $\deg g(x)=s-1$, $\m=2s$ and $g(x)$ generates a vector of length $s$ with all unit entries. Therefore, $l(x)$ also must generate a vector over $\Z_{2}$ with weight $s$. Hence, to generate such a cyclic one weight code we have two different cases; $r=s$ and $r>s$.

\bigskip

If $r=s$ then, to generate a vector with weight $s$, the degree of $l(x)$ must be $s-1$. So, $\left(l(x),g(x)+ua(x)\right)$ generates the codeword $(\underbrace{1\cdots 1}_\text{length $s$}| \underbrace{unit\cdots unit}_\text{length $s$})$.
\newline Further, if we multiply $\left(l(x),g(x)+ua(x)\right)$ by $h_{g}(x)$ we get $\left(h_{g}(x)l(x),uh_{g}(x)a(x)\right)$ and it generates codewords of order $2$. Since $r=s$ and the degrees of the polynomials $l(x)$ and $g(x)$ are $s-1$ we have $h_{g}=x+1$ and $h_{g}(x)l(x)=0$. Hence, $uh_{g}(x)a(x)$ must generate a vector with weight $2s$, i.e, $h_{g}(x)a(x)$ must generate a vector of length $s$ with all unit entries. This means that
\begin{eqnarray*}
(h_{g}(x)a(x)&=&\frac{x^{s}-1}{(x+1)}\\
\Rightarrow (x+1)a(x)&=&\frac{x^{s}-1}{(x+1)}\\
a(x)&=&\frac{x^{s}-1}{(x+1)^{2}}.
\end{eqnarray*}
 Hence we get $a(x)=\frac{x^{s}-1}{(x+1)^{2}}$. But, since we always assume that $s$ is an odd integer, $a(x)$ is not a factor of $(x^{s}-1)$ and this contradicts with the assumption $a(x)|(x^{s}-1)$. So, we can not allow $ua(x)h_{g}(x)$ to generate a vector, i.e, we must always choose $a(x)=g(x)$ to obtain $ua(x)h_{g}(x)=0$. So in the case where $\C$ is a one weight cyclic code generated by $\left(l(x),g(x)+ua(x)\right)$ in $R_{r=s,s}$, we only have $\C$ is a $\ZZ$-cyclic code of type $\left(s,s;0,1,0\right)$ with weight $\m=2s$.

\bigskip

 In the second case we have $r>s$. We know that $\C$ is a one weight cyclic code with weight $\m=2s$ and $g(x)=\frac{x^{s}-1}{x+1}$ generates a vector with exactly $s$ nonzero and all unit entries. Let $v=\left(v_{r},v_{s}\right)$ be a codeword of $\C$ such that $v_{r}=<l(x)>$ and $v_{s}=<g(x)+ua(x)>$. We can write $v$ as
$$
(\underbrace{a_{0}a_{1}\cdots a_{k-1}a_{k}}_\text{$s$ nonzero entries}| \underbrace{unit\cdots unit}_\text{length $s$})
$$
where $a_{i}\in \Z_{2}, k\in \Z$.
Since $\C$ is a $R$-submodule we can multiply $v$ by $u$, then we have
$$
(\underbrace{00\cdots 0}_\text{length $r$ }| \underbrace{u\cdots u}_\text{length $s$}).
$$
Let $w=\left(w_{r},w_{s}\right)$ be another codeword of $\C$ generated by $\left(h_{g}(x)l(x),ua(x)h_{g}(x)\right)$. Since $\C$ is a one weight code of weight $2s$, we can write $w=(\underbrace{b_{0}b_{1}b_{2}\cdots b_{t-1} b_{t}}_\text{$2s-2p$ nonzero entries }| \underbrace{u0uu0\cdots uu0u}_\text{$p$ nonzero entries}),$ $b_{i}\in \Z_{2}, t\in \Z$. Since $w+uv$ must be a codeword in $\C$, we have
$$
w+uv=(\underbrace{b_{0}b_{1}b_{2}\cdots b_{t-1}b_{t}}_\text{$2s-2p$ nonzero entries}| \underbrace{0u00u\cdots 00u0}_\text{$s-p$ nonzero entries}).
$$

Therefore, $wt\left(w+uv\right)=2s-2p+2s-2p=4s-4p$ and since $\C$ is a one weight code with $\m=2s$,
\begin{eqnarray*}
4s-4p=2s\Longrightarrow 2s=4p\Longrightarrow s=2p.
\end{eqnarray*}
But this contradicts with our assumption, that is, $s$ is an odd integer. Consequently, for $r>s$ and $g(x)\neq0$ there is no one weight $\ZZ$-cyclic code.
Under the light of all this discussion, we can give the following proved theorem.
\begin{theorem}
Let $\C$ be a $\ZZ$-cyclic code in $R_{r=s,s}$ generated by $\left(l(x),g(x)+ua(x)\right)$ with $\deg l(x)=\deg a(x)=\deg g(x)=s-1$. Then $\C$ is a one weight cyclic code of type $\left(r,s;0,1,0\right)$ with weight $\m=2s$. Furthermore, there do not exist any other one weight $\ZZ$-cyclic code with $g(x)\neq0$.
\end{theorem}

\begin{example}
Let $\C=\langle\left(l(x),g(x)+ua(x)\right)\rangle$ be a cyclic code in $R_{7,7}$ with $l(x)=g(x)=a(x)=\left(1+x+x^3\right) \left(1+x^2+x^3\right)=1+x+x^2+x^3+x^4+x^5+x^6$. Hence, $\C$ is a one weight code with weight $\m=14$ and the following generator matrix,
$$
\left(
\begin{array}{ccccccc|ccccccc}
 1 & 1 & 1 & 1 & 1 & 1 & 1  & 1+u & 1+u & 1+u & 1+u & 1+u & 1+u & 1+u
\end{array}
\right).$$
\end{example}

\section{Examples of One Weight $\ZZ$-cyclic Codes}

In this part of the paper, we give some examples of one weight $\ZZ$-cyclic codes. Furthermore, we look at their binary images under the Gray map that we defined in (\ref{graymap}), and we conclude that their binary images are optimal codes.
If the minimum distance of any code $\mathcal{C}$ get the possible maximum value according to its
length and dimension, then $\mathcal{C}$ is called optimal (distance-optimal)
or good parameter code. \newline Let $\C$ be a $\ZZ$-linear code with minimum distance $d=2t+1$, then we say $\C$ is a $t$-error correcting code. Since, the Gray map preserves the distances, $\Phi(\C)$ is also a $t$-error correcting code of length $r+2s$ over $\Z_{2}$. Since, $|\Phi(\C)|=|\C|$, we can write a sphere packing bound for a $\ZZ$-linear code $\C$. With the help of usual sphere packing bound in $\Z_{2}$,
\begin{eqnarray*}
|\Phi(\C)|\sum_{j=0}^{t}{r+2s\choose j}\leq|2^{r+2s}|,
\end{eqnarray*}
we have
\begin{eqnarray*}
|\C|\sum_{j=0}^{t}{r+2s\choose j}\leq|2^{r+2s}|=|\Z_{2}^{r}\times R^{s}|.
\end{eqnarray*}
If $\C$ attains the sphere packing bound above then it is called a \emph{perfect code}. Let $\C$ be a $\ZZ$-linear code of type $(3,2;2,1,0)$ with standard form of the generator matrix

$$
G=\left(
\begin{array}{ccc|cc}
 1 & 0 & 1 & 0 & u \\
0 & 1 & 0 & 0 & u \\ \hline
0 & 0 & 1 & 1 & 1+u
\end{array}
\right).$$
It is easy to check that $\C$ attains the sphere packing bound, so $\C$ is a perfect code. Moreover, the dual code $\C^{\perp}$ of $\C$ is generated by the matrix
\begin{eqnarray}\label{orn}
H=\left(
\begin{array}{ccc|cc}
 1 & 0 & 1 & u & 0 \\ \hline
1 & 1 & 0 & 1+u & 1
\end{array}
\right)
\end{eqnarray}
and $\C^{\perp}$ is a one weight $\ZZ$-linear code with weight $\m=4$.

\bigskip

Plotkin bound for a code over $F_{q}^{n}$ with the minimum distance $d$ is given by,
\begin{itemize}
\item[1.]
If $d=\left(1-\frac{1}{q}\right)n$, then $|\C|\leq 2qn$.
\item[2.]
If $d>\left(1-\frac{1}{q}\right)n$, then $|\C|\leq \frac{qd}{qd-(q-1)n}$.
\end{itemize}

If $\C\subseteq F_{q}^{n}$ attains the Plotkin bound then $\C$ is also an equidistant code \cite{Lint}. We can apply this result for a binary images of $\ZZ$-linear codes. If the Gray image $\Phi(\C)$ of a $\ZZ$-linear code $\C$ attains the Plotkin bound, then $\C$ is a one weight code as well. For an example, the $\ZZ$-linear code $\C^{\perp}$ given by the generator matrix (\ref{orn}) attains the Plotkin bound and also it is a one weight code.

Finally, we will give the following examples of one weight $\ZZ$-cyclic codes. We also determine the parameters of the binary images of these one weight cyclic codes. Further we list some of optimal binary codes in Table \ref{table} which we derived from one weight $\ZZ$-cyclic codes via the Gray map.

\begin{example}
Let $\C$ be a $\ZZ$-cyclic code in $R_{15,15}$ generated by $\left(l(x),g(x)+ua(x)\right)$ where
\begin{eqnarray*}
l(x)&=&1+x^3+x^4+x^6+x^8+x^9+x^{10}+x^{11},\\
g(x)&=&x^{15}-1,~a(x)=1+x^3+x^4+x^6+x^8+x^9+x^{10}+x^{11}.
\end{eqnarray*}
Then $\C$ is a one weight code with weight $\m=24$ and following generator matrix
$$
\resizebox{\linewidth}{!}{
$\displaystyle
\left(
\begin{array}{ccccccccccccccc|ccccccccccccccc}
 1 & 0 & 0 & 1 & 1 & 0 & 1 & 0 & 1 & 1 & 1 & 1 & 0 & 0 & 0 & u & 0 & 0 & u & u & 0 & u & 0 & u & u & u & u & 0 & 0 & 0 \\
 0 & 1 & 0 & 0 & 1 & 1 & 0 & 1 & 0 & 1 & 1 & 1 & 1 & 0 & 0 & 0 & u & 0 & 0 & u & u & 0 & u & 0 & u & u & u & u & 0 & 0 \\
 0 & 0 & 1 & 0 & 0 & 1 & 1 & 0 & 1 & 0 & 1 & 1 & 1 & 1 & 0 & 0 & 0 & u & 0 & 0 & u & u & 0 & u & 0 & u & u & u & u & 0 \\
 0 & 0 & 0 & 1 & 0 & 0 & 1 & 1 & 0 & 1 & 0 & 1 & 1 & 1 & 1 & 0 & 0 & 0 & u & 0 & 0 & u & u & 0 & u & 0 & u & u & u & u
\end{array}
\right)$}.$$
Furthermore, the binary image $\Phi(\C)$ of $\C$ is a $[45,4,24]$ code, which is a binary optimal code \cite{8}.
\end{example}

\begin{example}
The $\ZZ$-cyclic code $\C=\langle\left(l(x),g(x)+ua(x)\right)\rangle$ in $R_{9,9}$ is a one weight code with $\m=18$, where $l(x)=g(x)=a(x)=1+x+x^2+x^3+x^4+x^5+x^6+x^7+x^8$. $\C$ has the generator matrix of the form,
$$
\resizebox{\linewidth}{!}{
$\displaystyle
\left(
\begin{array}{ccccccccc|ccccccccc}
 1 & 1 & 1 & 1 & 1 & 1 & 1 & 1 & 1 & 1+u & 1+u & 1+u & 1+u & 1+u & 1+u & 1+u & 1+u & 1+u
\end{array}
\right).
$}
$$
The Gray image of $\C$ is an optimal $[27,2,18]$ binary linear code.
\end{example}

\begin{example}
Let $\C=\langle\left(l(x),g(x)+ua(x)\right)\rangle$, $l(x)=a(x)=1+x+x^2+x^4,~g(x)=x^7-1$, be a cyclic code in $R_{7,7}$. Then the generator matrix of $\C$ is
$$
\left(
\begin{array}{ccccccc|ccccccc}
 1 & 1 & 1 & 0 & 1 & 0 & 0 & u & u & u & 0 & u & 0 & 0 \\
 0 & 1 & 1 & 1 & 0 & 1 & 0 & 0 & u & u & u & 0 & u & 0 \\
 0 & 0 & 1 & 1 & 1 & 0 & 1 & 0 & 0 & u & u & u & 0 & u
\end{array}
\right).
$$

$\C$ is a one weight code with $\m=12$ and $\Phi(\C)$ is a $[21,3,12]$ binary code which is optimal.
\end{example}

\newpage

\begin{center}
\captionof{table}{The Table of Optimal Parameter Binary Linear Codes Derived from the One Weight $\ZZ$-cyclic Codes}\label{table}
\end{center}
\begin{tabularx}{\textwidth}{||X||c||c||}
\hline\hline
Generators & $\mathbb{Z}_{2}\mathbb{Z}_{2}[u]-$type & Binary Image  \\ \hline\hline
$l(x)=1+x+x^2+x^4,~g(x)=x^{21}-1,~a(x)=1+x+x^2+x^4+x^7+x^8+x^9+x^{11}+x^{14}+x^{15}+x^{16}+x^{18}$ & $
[7,21;0;0,3]$ & $[49,3,28]$ \\ \hline\hline
$l(x)=a(x)=1+x^2+x^4+x^5+x^6+x^8+x^9+x^{13}+x^{14}+x^{15}+x^{16}+x^{17}+x^{20}+x^{21}+x^{23}+x^{26},~g(x)=x^{31}-1
$ & $[31,31;0;0,5]$ & $[93,5,48]$  \\ \hline\hline
$l(x)=1+x+x^3+x^4+x^6+x^7+x^9+x^{10}+x^{12}+x^{13}+x^{15}+x^{16}+x^{18}+x^{19}+x^{21}+x^{22}+x^{24}+x^{25}~,g(x)=x^{15}-1,~a(x)=1+x+x^3+x^4+x^6+x^7+x^9+x^{10}+x^{12}+x^{13}$
& $[27,15;0;0,2]$ & $[57,2,38]$  \\ \hline\hline
$l(x)=1+x^2+x^3+x^4+x^7+x^9+x^{10}+x^{11}+x^{14}+x^{16}+x^{17}+x^{18}+x^{21}+x^{23}+x^{24}+x^{25}+x^{28}+x^{30}+x^{31}+x^{32},~g(x)=x^{21}-1,~a(x)=1+x^2+x^3+x^4+x^7+x^9+x^{10}+x^{11}+x^{14}+x^{16}+x^{17}+x^{18}$ & $[35,21;0,0,3]$ & $[77,3,44]$  \\ \hline\hline
\end{tabularx}

\section{Conclusion}

In this paper, we  study the one weight linear and cyclic codes over $\mathbb{Z}_{2}^{r}\times\left(\mathbb{Z}_{2}+u\mathbb{Z}_{2}\right)^{s}$ where $u^2=0$. We also classify this family of codes and present some illustrative examples. We also obtain optimal parameter binary codes derived from the Gray images of one weight $\ZZ$-cyclic codes.



\end{document}